\documentclass[preprint,11pt]{elsarticle}
\usepackage{amssymb}
\usepackage{url}
 \usepackage{listings}
\usepackage{xcolor}
\usepackage{hyperref}
\hypersetup{colorlinks,linkcolor={blue},citecolor={blue},urlcolor={red}}
\usepackage{amsmath}
\allowdisplaybreaks
\usepackage{minted}
\usepackage{algorithm}
\usepackage{color}
\usepackage[numbers]{natbib}
\usepackage{algpseudocode}
\usepackage{amsthm,setspace}

\newtheorem{theorem}{Theorem}
 \usepackage{lscape}
 \newtheorem{remark}{Remark}

\usepackage[margin=2.5cm]{geometry}
\usepackage{lineno}
\makeatletter
\def\ps@pprintTitle{%
 \let\@oddhead\@empty
 \let\@evenhead\@empty
 \def\@oddfoot{\reset@font\hfil\thepage\hfil}
 \let\@evenfoot\@oddfoot
}
\makeatletter
\begin{document}
\begin{frontmatter}

\title{A unified testing approach for log-symmetry using Fourier methods}

\author[label1]{Ganesh Vishnu Avhad}\ead{ma20d501@iittp.ac.in, avhadgv@gmail.com} 
\author[label3]{Sudheesh K. Kattumannil}  

\address[label1]{Department of Mathematics and Statistics, Indian Institute of Technology, Tirupati, India}
\address[label3]{Applied Statistics Unit, Indian Statistical Institute, Chennai, India}
  
\doublespacing
\begin{abstract}
Continuous and strictly positive data that exhibit skewness and outliers frequently arise in many applied disciplines. Log-symmetric distributions provide a flexible framework for modeling such data. In this article, we develop new goodness-of-fit tests for log-symmetric distributions based on a recent characterization. These tests utilize the characteristic function as a novel tool and are constructed using an $L^2$-type weighted distance measure. The asymptotic properties of the resulting test statistic are studied. The finite-sample performance of the proposed method is assessed via Monte Carlo simulations and compared with existing procedures. The results under a range of alternative distributions indicate superior empirical power, while the proposed test also exhibits substantial computational efficiency compared to existing methods. The methodology is further illustrated using real data sets to demonstrate practical applicability.
\end{abstract}

\end{frontmatter}
  \doublespacing

\section{Introduction}\label{sec1}
\noindent Log-symmetric distributions provide a flexible framework for modeling positive and right-skewed data and have found applications in economics, reliability, environmental studies, survival analysis, and finance. By assuming symmetry on the logarithmic scale, these distributions can accommodate skewness and diverse tail behaviors that are not well captured by classical symmetric models. In regression settings, log-symmetric models offer a robust alternative when normality assumptions on residuals are violated \citep{vanegas2017log, medeiros2017small, cunha2022log}. They have also been successfully applied to duration modeling and high-frequency financial data \citep{saulo2017log}, as well as to survival analysis and econometrics \citep{rubio2016survival, leiva2023bootstrap}. In financial applications, log-symmetric distributions yield improved risk estimation and option pricing compared to traditional models, see McCulloch \cite{mcculloch1978pricing}, Klebaner and Landsman \cite{klebaner2009option} and Bernard {et al}. \cite{bernard2017value}. A comprehensive discussion of their theoretical properties and practical relevance can be found in Ahmadi and Balakrishnan \cite{ahmadi2024characterizations} and the references therein.

A characterization of a distribution family is a property that uniquely identifies that family. For general discussions on characterization-based methods, see Nikitin \cite{nikitin2017tests}, and for recent results on log-symmetric distributions, see Ahmadi and Balakrishnan \cite{ahmadi2024characterizations}. Such characterizations provide a natural foundation for constructing goodness-of-fit tests, as they rely on fundamental properties that uniquely define the underlying distribution and often lead to high efficiency and good power. Owing to their theoretical advantages, including parameter-free formulations that facilitate testing composite hypotheses, characterization-based goodness-of-fit tests have gained increasing attention in recent years and have been developed for several distributions, including class of symmetric distributions, Pareto, gamma, Maxwell-Boltzmann, exponential and several other distributions. Comprehensive discussions and examples of such tests can be found in Henze et al. \cite{Henze2012}, Jovanovi\'{c} {et al}. \cite{JMNO:2015}, Milo\v{s}evi\'{c} and Obradovi\'{c} \cite{milovsevic2016characterization},  Partlett and Patil \cite{partlett2017measuring}, Ejsmont {et al}. \cite{ejsmont2023test}, Avhad {et al}. \cite{avhad2024}, Rano\v{s}ov\'{a} and Hlubinka \cite{ranovsova2025symmetry} and Avhad and Ebner\cite{avhad2025efficient}. For recent developments and discussions of Fourier-based tests, we refer the reader to Allison and Pretorius \cite{allison2017monte} and Ndwandwe {et al}. \cite{ndwandwe2023new}, and the references therein.

Although accurate modeling of positively skewed data is essential in many applications, goodness-of-fit testing for log-symmetric distributions has received relatively limited attention. Only a few recent studies address this problem. For instance, Anjana and Kattumannil \cite{Anjana21072025} proposed a jackknife empirical likelihood ratio test based on probability weighted moments, but this method exhibits some limitations in small sample settings, as it fails to achieve the nominal significance levels. More recently, characterization-based empirical likelihood tests involving order statistics were developed in Avhad and Lahiri \cite{avhad2025jackknife} and Avhad {et al}. \cite{avhad2025family}, though these procedures are computationally demanding. {However, existing methods are constructed using integral-type functionals within $U$-statistic and empirical likelihood frameworks. In contrast, the present work introduces a class of goodness-of-fit tests based on a weighted $L^2$-distance between characteristic functions, motivated by a recent characterization of log-symmetry. To the best of our knowledge, no such $L^2$-distance-based test has been proposed for log-symmetric distributions. The Fourier-based approach yields explicit test statistics for common weight functions, is consistent against a wide range of alternatives under mild conditions, and avoids resampling and constrained optimization, making the method both theoretically tractable and computationally efficient.
}

The rest of the article is organized as follows. In Section \ref{sec2}, we present the log-symmetric distributional characterization theorem, and we propose a testing procedure for assessing the fit to the family of log-symmetric distributions. We also examine the theoretical results and asymptotic properties of the test statistics. In Section \ref{sec4}, we conduct extensive Monte Carlo simulation studies to evaluate the finite sample performance of the proposed methods. The practical applicability of the proposed methods is illustrated using real-life datasets in Section \ref{sec5}, followed by conclusions in Section \ref{sec6}.

\section{Test statistic based on a characterization}\label{sec2}
\noindent
 Let $X, X_1, X_2,\ldots, X_n$ be independent and identically distributed (i.i.d.) random variables with continuous distribution function (d.f.) $F_X$ and density function $f_X$ with support on \( \mathbb{R}^+ \). We consider the problem of testing the null hypothesis of log-symmetry  
\begin{equation}\label{hypo}
    H_0: F_X(x)= 1-F_X(1/x), \quad \text{for all} \quad x>0,
\end{equation}
against the general alternative. 

\subsection{Characterization of log-symmetry}
\noindent We consider the following characterization result of log-symmetric distributions established in Theorem 4.1 of Ahmadi and Balakrishnan \cite{ahmadi2024characterizations} to develop the test. The detailed proof is provided in the cited work. Denote \( X_{(1)}, X_{(2)}, \dots, X_{(n)} \) are the order statistics of a random sample of size $n$ drawn from $F_X$.  

\begin{theorem}[Ahmadi and Balakrishnan, 2024]\label{thm1}
 Let $X, X_1, \dots, X_n$  be a random sample from a distribution $F_X$  with support on $\mathbb{R}^+$. For a fixed $i$ where $1 \leq i \leq n/2$ and $n > 1$, if the statistics $X_{(i)}$ and  $1/X_{(n-i+1)}$  are identically distributed, then $X$ follows a log-symmetric distribution. 
\end{theorem}
Based on the above characterization with $i=1$ and $n=k(<n)$, Avhad {et al}. \cite{avhad2025family} proposed a $U$-empirical process-based tests for log-symmetric distributions with the following distribution-free test statistic: 
\begin{align}\label{delta1}
  \Delta_n &= \int_{0}^\infty \big( K_n(t) - R_n(t)  \big) dF_n(t), 
\end{align}
where  
\begin{equation*}
    K_n(t) = \binom{n}{k}^{-1} \sum_{{\pi}_k}  {I}\big\{\min \{X_{i_1}, X_{i_2},\ldots, X_{i_k}\} \leq t \big\},  
\end{equation*}
\begin{equation*}
    R_n(t) = \binom{n}{k}^{-1} \sum_{{\pi}_k}  {I}\big\{\max \{X_{i_1},  X_{i_2},\ldots, X_{i_k}\} \geq 1/t \big\},  
\end{equation*}
where ${\pi}_k= \{(i_1, i_2,\ldots, i_k): 1\leq i_1< i_2<\ldots< i_k \leq n \}$ and ${I}(\cdot)$ denote the indicator function. Based on the asymptotic theory of $U$-statistics, Avhad {et al}. \cite{avhad2025family} derived the asymptotic distribution of $\Delta_n$. However, due to the complexity of deriving the asymptotic null variance of $\Delta_n$, they adopted a nonparametric simulated critical region approach to approximate the critical values. Subsequently, empirical likelihood-based tests were proposed. However, in the mentioned paper, the use of order statistics in constructing the test statistic leads to a substantial computational burden in the simulation study.
As an alternative, a more efficient test can be developed based on the same characterization by setting $i=1$ and using the characteristic function $\phi_X$ instead of the d.f. $F_X$.  Thus, we have the following result.   
\begin{theorem}\label{newthm1}
Let $X, X_1, X_2, \ldots, X_n$ be a non-negative i.i.d. random sample from a d.f. $F_X$. Then $X_{(1)}$ and $1/{X_{(n)}}$ are identically distributed if and only if $F_X$ is log-symmetric. 
\end{theorem}

{Note that, Theorem \ref{thm1} provides a characterization of log-symmetry through the distributional equivalence of suitably chosen order statistics. To develop a tractable testing procedure based on this characterization, it is convenient to reformulate this distributional equality in an analytic form. One natural approach is to express the equality in distribution in terms of characteristic functions, which allows us to translate the order-statistic-based characterization into a Fourier-domain representation  \citep{allison2017monte,ebner2024cauchy}. This Fourier-based representation provides an equivalent but more convenient formulation for constructing the proposed test statistics.}


\subsection{Test construction}
\noindent Based on the characterization described in Theorem \ref{newthm1} and in the spirit of tests based on the characteristic function \citep{huvskova2012tests}, to test the hypothesis stated in \eqref{hypo}, we propose a class of tests based on a weighted $L^2$-distance measure as
\begin{equation}
     \Delta_{n,a}=  \int_{0}^\infty \Big|\phi_{X_{(1)}}(t)- \phi_{{1}/{X_{(n)}}}(t) \Big|^2 w_a(t)dt,
 \end{equation}
where 
\begin{equation*}
    \phi_{X_{(1)}}(t)= \mathbb{E}\!\left(e^{itX_{(1)}} \right)
    \quad \text{and} \quad
    \phi_{1/X_{(n)}}(t)= \mathbb{E}\!\left(e^{it/X_{(n)}} \right)
\end{equation*}
{denote the characteristic functions of $X_{(1)}$ and $1/X_{(n)}$, respectively} and $w_a(\cdot)$ is the weight function. Consider 
\begin{align*}
    D_n(t)&=  \mathbb{E}\Big(e^{itX_{(1)}} \Big)- \mathbb{E}\Big(e^{ {it}/{X_{(n)}}} \Big) \\
    &= n\int_{0}^\infty e^{itx} (1-F(x))^{n-1} dF(x)- n\int_{0}^\infty e^{it/x} F(x)^{n-1} dF(x).
\end{align*}
We propose the following test statistic to test the hypothesis:
\begin{equation}\label{Tna}
    T_{n,a}= n\int_{0}^\infty |\widehat{D}_n(t)|^2 w_a(t)dt, 
\end{equation}
 where
 \begin{small}
\begin{align*}
    \widehat{D}_n(t)&= \frac{1}{n}\sum_{m=1}^n \left(\frac{n-m}{n-1} \right)^{n-1} \left\{\cos(tX_{(m)})+i\sin (tX_{(m)})\right\}  - \frac{1}{n}\sum_{l=1}^n \left(\frac{l-1}{n-1} \right)^{n-1} \left\{\cos \left(\frac{t}{X_{(l)}}\right)+i\sin \left(\frac{t}{X_{(l)}}\right) \right\}\\
    &= \frac{1}{n}\sum_{m=1}^n \beta_{m,n}\left\{\cos(tX_{(m)})+i\sin (tX_{(m)})\right\}-\frac{1}{n}\sum_{l=1}^n \gamma_{l,n}\left\{\cos \Big(\frac{t}{X_{(l)}}\Big)+i\sin \left(\frac{t}{X_{(l)}}\right) \right\} \\
    & =: A_n(t) - B_n(t). 
\end{align*}
\end{small}
{Now, consider
\[
|\widehat{D}_n(t)|^2
= \left(A_n(t)-B_n(t)\right)\overline{\left(A_n(t)-B_n(t)\right)}
= |A_n(t)|^2 + |B_n(t)|^2 - 2\Re\left(A_n(t)\overline{B_n(t)}\right),
\]
where, $\Re(\cdot)$ denotes the real part of a complex number.
Using the identity
$\cos(a-b)=\cos (a) \cos (b) + \sin (a) \sin (b) $,
we obtain
\[
|A_n(t)|^2
= \frac{1}{n^2}\sum_{m=1}^n\sum_{l=1}^n
\beta_{m,n}\beta_{l,n}
\cos\big(t(X_{(m)}-X_{(l)})\big).
\]
Similarly,
\[
|B_n(t)|^2
= \frac{1}{n^2}\sum_{m=1}^n\sum_{l=1}^n
\gamma_{m,n}\gamma_{l,n}
\cos\Big(t\Big(\tfrac{1}{X_{(m)}}-\tfrac{1}{X_{(l)}}\Big)\Big), 
\]
and  
\[
\Re\{A_n(t)\overline{B_n(t)}\}
= \frac{1}{n^2}\sum_{m=1}^n\sum_{l=1}^n
\beta_{m,n}\gamma_{l,n}
\cos\Big(t\Big(X_{(m)}-\tfrac{1}{X_{(l)}}\Big)\Big).
\]
Substituting these expressions into the definition of $T_{n,a}$ and interchanging summation and integration yield
\begin{align}\label{Dn}
T_{n,a}
&= \frac{1}{n}\sum_{m=1}^n\sum_{l=1}^n \beta_{m,n}\beta_{l,n}
\int_0^\infty \cos\big(t(X_{(m)}-X_{(l)})\big) w_a(t)\,dt \nonumber\\
&\quad + \frac{1}{n}\sum_{m=1}^n\sum_{l=1}^n \gamma_{m,n}\gamma_{l,n}
\int_0^\infty \cos\Big(t\Big(\tfrac{1}{X_{(m)}}-\tfrac{1}{X_{(l)}}\Big)\Big) w_a(t)\,dt \nonumber\\
&\quad - \frac{2}{n}\sum_{m=1}^n\sum_{l=1}^n \beta_{m,n}\gamma_{l,n}
\int_0^\infty \cos\Big(t\Big(X_{(m)}-\tfrac{1}{X_{(l)}}\Big)\Big) w_a(t)\,dt,
\end{align}
and we define
\begin{equation*}
   {I}_{w,a}(b) = \int_{0}^{\infty} \cos(b t)w_a(t)dt.
\end{equation*}}
The selection of the weight function $w_a(\cdot)$ has been extensively explored in the literature. Two widely used forms are $w_a(t) = e^{-a|t|}$ and $w_a(t) = e^{-at^2}$. These functions correspond to kernel-based approaches, where $e^{-a|t|}$ represents a scaled Laplace kernel with bandwidth $1/a$, and $e^{-at^2}$ corresponds to a scaled Gaussian kernel with bandwidth $1/(a\sqrt{2})$. For these weight functions, the resulting integrals are given by
\begin{equation*}
 \int_{0}^{\infty} \cos(b t) e^{-a |t|} dt = \frac{a}{a^2 + b^2} \quad \text{and} \quad \int_{0}^{\infty} \cos(bt) e^{-at^2} dt = \frac{1}{2} \sqrt{\frac{\pi}{a}}\, e^{- \frac{b^2}{4a}}, \quad \text{for } a > 0,\, b \in \mathbb{R}.   
\end{equation*}
Hence, for $w_a(t) = e^{-a|t|}$, we obtain the test statistic  
\begin{align*}
T_{n,a}^{(1)} &= \dfrac{1}{n} \sum_{m=1}^{n}\sum_{l=1}^{n} \beta_{m,n} \beta_{l,n} \frac{a}{a^2 + (X_{(m)} - X_{(l)})^2}
+ \dfrac{1}{n} \sum_{m=1}^{n}\sum_{l=1}^{n}  \gamma_{m,n} \gamma_{l,n} \frac{a}{a^2 + \left(\frac{1}{X_{(m)}} - \frac{1}{X_{(l)}}\right)^2}  \\
&\hspace{0.5cm}- \dfrac{2}{n}\sum_{m=1}^{n}\sum_{l=1}^{n} \beta_{m,n} \gamma_{l,n} \frac{a}{a^2 + \left(X_{(m)} - \frac{1}{X_{(l)}}\right)^2}, 
\end{align*}  
and for $w_a(t) = e^{-at^2}$ the test statistic becomes
\begin{align*}
 T_{n,a}^{(2)} &= \dfrac{1}{n} \sqrt{\dfrac{\pi}{a}} \sum_{m=1}^{n}\sum_{l=1}^{n} \beta_{m,n} \beta_{l,n}  e^{-\frac{1}{4a}(X_{(m)} - X_{(l)})^2}
+ \dfrac{1}{n} \sqrt{\dfrac{\pi}{a}}\sum_{m=1}^{n}\sum_{l=1}^{n}  \gamma_{m,n} \gamma_{l,n} e^{-\frac{1}{4a}\left(\frac{1}{X_{(m)}} - \frac{1}{X_{(l)}}\right)^2}\\
&\hspace{0.5cm}- \dfrac{2}{n}\sqrt{\dfrac{\pi}{a}} \sum_{m=1}^{n}\sum_{l=1}^{n} \beta_{m,n} \gamma_{l,n}  e^{-\frac{1}{4a} \left(X_{(m)} - \frac{1}{X_{(l)}}\right)^2}.
\end{align*} 
As both tests $T_{n, a}^{(1)}$ and $T_{n, a}^{(2)}$ are constructed using a weighted $L^2$-distance measure, large values of these test statistics provide stronger evidence against the null hypothesis.
 
\subsection{Asymptotic theory}
\subsubsection{The limit distribution under the null hypothesis}\label{nulldist}
\noindent Under the null hypothesis that the data follow a log-symmetric distribution, we investigate the asymptotic behavior of $T_{n, a}$ as $n \to \infty$. Define the empirical processes
\begin{equation*}
\widehat{Z}_n(s) =\dfrac{1}{n} \sum_{j=1}^n \beta_{j,n} \cos(s X_{(j)}), \quad 
\widehat{Z}_n^*(s) = \dfrac{1}{n}\sum_{j=1}^n \gamma_{j,n} \cos\left( \frac{s}{X_{(j)}} \right), \quad s > 0.
\end{equation*}
Then, the statistic \eqref{Dn} can be expressed as
\begin{equation*}
T_{n,a} = \int_0^\infty \left( \widehat{Z}_n(s)^2 - \widehat{Z}_n^*(s)^2 \right) w_a(s) \, ds.
\end{equation*}   
The statistic $T_{n, a}$ is constructed as a weighted $L^2$-type functional based on the empirical characteristic function. The asymptotic property under the null distribution for such a class of statistics has been investigated in several previous studies; see e.g., Klar and Meintanis \cite{klar2005tests}, Baringhaus {et al}. \cite{baringhaus2017limit} and Ebner and Henze \cite{ebner2020tests}. {The asymptotic results for the test statistic $T_{n,a}$ depend on standard theory for empirical characteristic function-based processes. To ensure weak convergence and consistency, we impose the following mild regularity conditions:
\begin{enumerate}
    \item[\textbf{(C1)}] The random variables $X_1,\dots,X_n$ are i.i.d. and take strictly positive values.
    
    \item[\textbf{(C2)}] The logarithm of the observations, $Y_i = \log X_i$, has finite second moment, i.e., $\mathbb{E}[Y_i^2] < \infty$. 
    
    \item[\textbf{(C3)}] The weight function $w_a(t)$ is non-negative, integrable over $[0,\infty)$, and square-integrable, i.e.,
    \[
        \int_0^\infty w_a(t) \, dt < \infty \quad \text{and} \quad \int_0^\infty w_a^2(t) \, dt < \infty.
    \]
    
    \item[\textbf{(C4)}] The kernel function $I_{w,a}(\cdot)$ used to express $T_{n,a}$ satisfies
    \(\int_0^\infty I_{w,a}^2(b) \, db < \infty,\)
    which guarantees that the quadratic forms in the test statistic converge.
\end{enumerate}}
Under this regularity conditions, we have the following result. 
\begin{theorem}
Let $X_1, \ldots, X_n$ be i.i.d. random variables from a log-symmetric distribution. Then, under the null hypothesis, the test statistic $T_{n,a}$ satisfies
\begin{equation*}
T_{n,a} \xrightarrow{d} \int_0^\infty \bigl( Z(s) - Z^*(s) \bigr)^2 \, w_a(s)\, ds, 
\quad \text{as } n \to \infty,
\end{equation*}
{where $(Z(s), Z^*(s))_{s \ge 0}$ is a jointly Gaussian process with zero mean. The marginal processes $Z(s)$ and $Z^*(s)$ are identically distributed with covariance function $\mathbb{E}[Z(s)Z(t)] = K(s,t)$, and their joint covariance structure, including cross-covariances, is induced by the transformation $x \mapsto 1/x$ under the log-symmetry assumption.} The kernel $K(s,t)$ depends on the underlying log-symmetric distribution and the choice of the weight function through $I_{w,a}$.
\end{theorem}
\begin{proof}
The result follows by expressing $T_{n, a}$ as a difference of two quadratic forms in Hilbert space 
\begin{equation*}
    T_{n,a} = \|\widehat{Z}_n\|_{\mathbb{H}_a}^2 - \|\widehat{Z}_n^*\|_{\mathbb{H}_a}^2,
\end{equation*}
where $\mathbb{H}_a = L^2_w(0, \infty)$ is a weighted $L^2$ space with norm
\begin{equation*}
\|f\|_{\mathbb{H}_a}^2 = \int_0^\infty f(s)^2 w_a(s) \, ds.
\end{equation*}
Under the null hypothesis of log-symmetry, the processes $\widehat{Z}_n$ and $\widehat{Z}_n^*$ converge weakly to identically and independently distributed standard Gaussian processes $Z$ and $Z^*$, respectively, in $\mathbb{H}_a$. The convergence in distribution of $T_{n, a}$ then follows from the continuous mapping theorem.
\end{proof}
 
The limiting distribution is centered and nondegenerate, depending on the properties of the underlying log-symmetric distribution and the choice of kernel and weight function. The asymptotic null distribution is not in closed form, and thus it is hardly useful to find the critical values of the test statistic $T_{n, a}$. In practice, the critical values of the test statistic can be effectively approximated using a resampling technique such as the warp-speed bootstrap.

\subsubsection{Consistency of the test}\label{cons}

{\noindent Under the alternative hypothesis that $X$ is positive and $\log X$ is not symmetric, i.e., $X \in \mathcal{A} = \{ X>0 : \log X \text{ is not symmetric} \}$,} we present the limiting distribution of $T_{n, a}$ in the following result. 

\begin{theorem}
Under the alternatives stated above, we have 
 \begin{equation*}
     \dfrac{T_{n,a}}{n} \xrightarrow{p}\Delta_{n,a}:= \int_{0}^\infty \Big|\phi_{X_{(1)}}(t)- \phi_{\frac{1}{X_{(n)}}}(t) \Big|^2 w_a(t)dt,
 \end{equation*}
 as $n\to \infty$, with $\Delta_{n,a}=0$ if and only if $X$ follows log-symmetric distribution. 
\end{theorem}
\begin{proof}
As defined in \eqref{Tna}, the proposed test statistic is defined as
 \begin{equation*}
\dfrac{T_{n,a}}{n} = \int_{0}^{\infty} \left| \widehat{D}_n(t) \right|^2 w_a(t)\,dt,
\end{equation*}
where
 \begin{equation*}
\widehat{D}_n(t) = \sum_{m=1}^n \beta_{m,n} \left[ \cos(tX_{(m)}) + i \sin(tX_{(m)}) \right] - \sum_{l=1}^n \gamma_{l,n} \left[ \cos\left(\frac{t}{X_{(l)}}\right) + i \sin\left(\frac{t}{X_{(l)}}\right) \right],
\end{equation*}
with $\{z_{m,n}\}_{m=1}^n$ being deterministic weights that satisfy $\sum_{m=1}^n z_{m,n} = 1$.
Under regularity conditions (such as bounded support or suitable tail decay), and by standard results from empirical process theory, we have pointwise convergence in probability (Henze \cite{H:2024}, Theorem $6.17$)  
\begin{equation*}
\widehat{D}_n(t) \xrightarrow{p} D(t) := \phi_X(t) - \phi_{1/X}(t),
\end{equation*}
as $n \to \infty$.
Assuming the weight function $w_a(t)$ satisfies $\int_{0}^\infty |D(t)|^2 w_a(t)\,dt < \infty$ (which holds for both $w_a(t) = e^{-a|t|}$ and $w_a(t) = e^{-at^2}$), we apply the dominated convergence theorem to get
 \begin{equation*}
\frac{T_{n,a}}{n} = \int_{0}^{\infty} \left| \widehat{D}_n(t) \right|^2 w_a(t)\,dt \xrightarrow{p} \int_{0}^\infty \left| \phi_X(t) - \phi_{1/X}(t) \right|^2 w_a(t)\,dt =: \Delta_{a}.
\end{equation*}
 Under the alternative hypothesis, we have
\[
\frac{T_{n,a}}{n} \xrightarrow{p} \Delta_a > 0.
\]
Consequently, 
\[
T_{n,a} = n \Delta_a + o_p(n) \to \infty \quad \text{in probability as } n \to \infty,
\]
which confirms that the proposed test consistently detects departures from log-symmetry. To complete the proof, note that $\phi_X(t) = \phi_{1/X}(t)$ for all $t$ if and only if $X \overset{d}{=} 1/X$, which characterizes log-symmetry. Consequently, $\Delta_a = 0$ if and only if $X$ is log-symmetric. Hence, the proposed test is consistent. 
\end{proof}
{
\begin{remark}
The class of alternatives consists of all positive distributions for which $\log X$ is not symmetric about its median. Under such alternatives, the test statistic grows linearly with $n$, i.e., $T_{n,a} = n \Delta_a + o_p(n)$ with $\Delta_a > 0$, which ensures the consistency of the proposed test. Under the null hypothesis of log-symmetry, we have $\Delta_a = 0$, and $T_{n,a}$ converges in distribution to its limiting null distribution.
\end{remark}}

\section{Simulation study}\label{sec4}
\noindent In this section, we carried out a comprehensive Monte Carlo simulation study to assess the performance of the newly proposed methods $T_{n, a}^{(1)}$ and $T_{n, a}^{(2)}$ with different values of the tuning parameter $a$. All calculations were performed in \texttt{R} software \citep{team2020ra}. Given the lack of a closed-form expression for the null distribution of the proposed test statistics $T_{n, a}^{(1)}$ and $T_{n, a}^{(2)}$, the critical values are obtained through Monte Carlo simulation with 10000 replications. To improve computational efficiency, the warp-speed bootstrap approach proposed by Giacomini {et al}. \cite{giacomini2013warp} is used to estimate the empirical type I error rates and powers. 

The performance of these tests is compared with existing methods, such as recent methods introduced by Anjana and Kattumannil \cite{Anjana21072025}, Avhad and Lahiri \cite{avhad2025jackknife} and Avhad {et al}. \cite{avhad2025family}. Let $X_{(i)}$ denote the $i$-th order statistics from the random sample $X_1, X_2, \ldots, X_n$ drawn from $F_X$. The corresponding test statistics are defined as follows: 
\begin{enumerate}
    \item Anjana and Kattumannil \cite{Anjana21072025} proposed a jackknife empirical likelihood (JEL) ratio test ($\widehat{\Delta}^J_{W}$) based on the probability-weighted moments. The test statistic is given by
    \begin{equation*}
        \widehat{\Delta}_{W}= \binom{n}{\beta+1}^{-1}\dfrac{1}{\beta+1}\sum_{i=1}^n \bigg[\binom{i-1}{\beta}X_{(i)}- \binom{n-i}{\beta} \dfrac{1}{X_{(i)}}\bigg],
    \end{equation*}
    where $\beta=3$ is chosen for our simulation purposes.
    \item Avhad and Lahiri \cite{avhad2025jackknife} proposed a goodness-of-fit test based on the ratio of two consecutive order statistics and subsequently developed two empirical likelihood-based methods: the JEL ratio test, denoted by $\widehat{\Delta}^{J}_{R}$, and the adjusted JEL (AJEL) ratio test, denoted by $\widehat{\Delta}^{AJ}_{R}$. The test statistic is given by
    \begin{equation*}
        \widehat{\Delta}_{R}= \dfrac{1}{\binom{n}{4}}  \sum_{\substack{i,j,k,l=1 \\ i<j<k<l}}^n \left[{I}\left\{ \dfrac{X_{(2);X_i, X_j, X_k}}{X_{(1);X_i, X_j, X_k}} \leq X_l \right\} - {I}\left\{ \dfrac{X_{(3);X_i, X_j, X_k}}{X_{(2);X_i, X_j, X_k}}\leq X_l \right\}\right].
    \end{equation*}
    
\item Avhad {et al}. \cite{avhad2025family} proposed a goodness-of-fit test based on a characterization of log-symmetry and developed three additional tests: one test is constructed using the asymptotic critical region, denoted by $\widehat{\Delta}^{S}_{D}$, while the other two are based on empirical likelihood approach, such as JEL ratio test $(\widehat{\Delta}^{J}_{D})$ and the AJEL ratio test $(\widehat{\Delta}^{AJ}_{D})$. The $U$-statistic-based estimator is given by
\begin{align*}\label{In}
    \widehat{{\Delta}}_{D}& = \dfrac{1}{ \binom{n}{k+1}}  \sum_{\mathcal{L}_{k+1}}  \tau(X_{i_1}, X_{i_2},\ldots, X_{i_{k+1}}),
\end{align*}
with the associated symmetric kernel  
\begin{align*}
    \tau(X_1, X_2,\ldots, X_{k+1})
    &=\dfrac{1}{(k+1)!} \sum_{\pi \in \prod(k+1)} \Big[  {I}\big\{\min (X_{\pi(1)}, \ldots, X_{\pi(k)}) \leq X_{\pi({k+1})} \big\} \\
   &\hspace{3.3cm} - {I}\big\{\max (X_{\pi(1)}, \ldots, X_{\pi(k)}) \geq 1/X_{\pi({k+1})} \big\}\Big],
\end{align*}
where $\tau(\cdot)$ is the kernel function and $\prod(n)$ is the set of permutations of the set $\{1, 2,\ldots, n\}$. We implement the test for $k=3$ in order to obtain the numerical results presented.
\end{enumerate}

\noindent For each method, we computed the empirical type I error rate and the power of the tests at $5\%$ and $1\%$ significance levels, with sample sizes $n = 10$, $25$ and $50$. To find the empirical type I error rate, we drawn samples from several log-symmetric distributions, including the log-normal ${N_L}(0, 1)$, log-Laplace ${Lap_L}(0,1)$, log-logistic ${Log_L}(0, 1)$, log-t ${t_L}(5)$, and log-Cauchy ${C}_L(0,1)$ distributions. Empirical type I error rates are computed using the warp-speed bootstrap method, where samples are generated from the assumed log-symmetric distributions. The results of empirical type I error rates for different choices of the tuning parameter $a$ are reported in Table~\ref{tablep1}. From Table \ref{tablep1}, we see that for small sample sizes, the empirical type I error rates of all methods are reasonably close to the corresponding nominal significance levels, except the $\widehat{\Delta}_{W}^J$ test, which tends to deviate noticeably. The performance of the $\widehat{\Delta}_{W}^J$ test improves with increasing sample size, approaching the nominal level of significance in large-sample settings.

To assess the empirical powers, the alternative distributions considered were the gamma (G), inverse gamma (IG), chi-square $(\chi^2)$, Levy (L), Weibull (W), Maxwell (M), Pareto (P), inverse beta (IB), Benini (B), and tilted Pareto (TP) distributions. {This distributions were chosen as alternatives because they are common lifetime models whose log-transforms are asymmetric, thereby representing natural and practically relevant departures from log-symmetry.}  For empirical power calculations, bootstrap samples are drawn from the standard log-normal distribution. As the samples are generated under the null model, the particular choice of log-symmetric distribution has negligible influence on the inference. This observation is especially pertinent because the log-symmetric class extends the log-normal distribution and captures a wide spectrum of positively skewed distributions.
To evaluate performance, we present the empirical power results for all proposed tests under various choices of the tuning parameter: $a =0.5, 1, 1.5$ and $3$. Tables \ref{tablep2} and \ref{tablep3} contain these empirical power estimates. The results shown in Tables \ref{tablep2} and \ref{tablep3} indicate that the overall performance of the newly proposed tests demonstrates very good or at least comparable power with existing methods across nearly all alternative distributions.  A comparison of the power results for different values of the tuning parameter indicates that the test statistics $T_{n, a}^{(1)}$ and $T_{n, a}^{(2)}$ generally yield the highest empirical power across the considered alternatives. {However, based on the simulation results, we recommend using the test with a setting parameter $a = 1$, as it consistently demonstrates favorable performance in various scenarios. To clearly summarize the simulation results and improve readability, the maximum empirical power in each row of the tables is highlighted in bold.} For the inverse gamma and chi-square alternatives, the empirical likelihood-based tests exhibits better power performance compared to the proposed methods. It is also observed that the $\widehat{\Delta}_{W}^J$ test exhibits slightly better performance than the proposed tests under certain alternatives, making it a noteworthy competitor in those cases. Overall, for small sample sizes, the newly proposed tests demonstrate higher power performance and attain the nominal significance level compared to existing methods.

\begin{landscape} 
\begin{table}[!ht]
 \centering
  \caption{Empirical type I error rates at $5\%$ and $1\%$ level of significance }
    \resizebox{21 cm}{!}{
    \fontsize{8pt}{11pt}\selectfont
    \begin{tabular}{rrrrrrrrrrrrrrrrrrrrrrrrrrr}
    \hline
Alt. & $\alpha$& $n$& $T_{n,0.5}^{(1)}$ & $T_{n,1}^{(1)}$ & $T_{n,1.5}^{(1)}$ & $T_{n,3}^{(1)}$ &$T_{n,0.5}^{(2)}$ & $T_{n,1}^{(2)}$&$T_{n,1.5}^{(2)}$ & $T_{n,3}^{(2)}$&$\widehat{\Delta}^J_{W}$  & $\widehat{\Delta}^J_{R}$ & $\widehat{\Delta}^{AJ}_{R}$&  $\widehat{\Delta}^S_{D}$&$\widehat{\Delta}^J_{D}$& $\widehat{\Delta}^{AJ}_{D}$ \\[1.5ex]
 \hline  
$N_L(0,1)$&0.05&10&0.0526&0.0512&0.0518&0.0487&0.0505&0.0499&0.0495&0.0474&0.1420&0.0425&0.0436&0.0487&0.0512&0.0497\\[0.5ex] 
&&25&0.0531&0.0522&0.0522&0.0495&0.0547&0.0543&0.0532&0.0503&0.1161&0.0474&0.0430&0.0489&0.0465&0.0464 \\[0.5ex] 
&&50&0.0507&0.0483&0.0478&0.0465&0.0487&0.0461&0.0465&0.0470&0.0870&0.0474&0.0430&0.0489&0.0465&0.0464 \\[1.5ex]  
\hline 
&0.01&10&0.0099&0.0095&0.0086&0.0075&0.0097&0.0069&0.0077&0.0076&0.1180&0.0119&0.0106&0.0071&0.0112&0.0089   \\[0.5ex] 
&&25&0.0111&0.0105&0.0100&0.0090&0.0106&0.0097&0.0093&0.0095&0.0630&0.1088&0.0077&0.0111&0.0102& 0.0086 \\[0.5ex]    
&&50&0.0106&0.0084&0.0079&0.0087&0.0083&0.0083&0.0081&0.0076&0.0380&0.0835&0.0099&0.0071&0.0107& 0.0108 \\[1.5ex]   
\hline 
$Lap_L(0,1)$&0.05&10&0.0504&0.0474&0.0466&0.0442&0.0479&0.0473&0.0457&0.0451&0.1690&0.0504& 0.0518& 0.0499& 0.0516& 0.0450 \\[0.5ex]
&&25&0.0527&0.0555&0.0561&0.0544&0.0547&0.0550&0.0562&0.0552&0.1141&0.0505& 0.0488& 0.0452& 0.0510& 0.0435 \\[0.5ex]
&&50&0.0503&0.0506&0.0511&0.0494&0.0489&0.0493&0.0500&0.0507&0.0651&0.0497& 0.0452& 0.0456& 0.0437& 0.0438 \\[1.5ex] 
\hline     
&0.01&10&0.0067&0.0068&0.0066&0.0087&0.0065&0.0069&0.0081&0.0093&0.1890&0.0119& 0.0081& 0.0086& 0.0074& 0.0097 \\[0.5ex]
&&25&0.0079&0.0078&0.0084&0.0096&0.0080&0.0082&0.0086&0.0091&0.1171&0.0114&0.0087& 0.0114& 0.0096& 0.0085 \\[0.5ex]  
&&50&0.0097&0.0107&0.0110&0.0116&0.0094&0.0102&0.0113&0.0113&0.0540&0.0101&0.0108& 0.0115& 0.0113& 0.0085 \\[1.5ex] 
\hline  
$Log_L(0,1)$&0.05&10&0.0542&0.0530&0.0524&0.0533&0.0534&0.0517&0.0532&0.0523&0.1820&0.0508& 0.0482& 0.0514& 0.0445& 0.0462 \\[0.5ex] 
&&25&0.0451&0.0458&0.0453&0.0451&0.0483&0.0493&0.0473&0.0455&0.0810&0.0501& 0.0484& 0.0460& 0.0510& 0.0443 \\[0.5ex]
&&50&0.0511&0.0529&0.0546&0.0551&0.0514&0.0511&0.0510&0.0525&0.0671&0.0513& 0.0502& 0.0457& 0.0494& 0.0440 \\[1.5ex] 
\hline  
&0.01&10&0.0086&0.0087&0.0092&0.0101&0.0087&0.0092&0.0102&0.0102&0.1551&0.0111& 0.0092& 0.0109& 0.0112& 0.0077 \\[0.5ex] 
&&25&0.0128&0.0116&0.0116&0.0110&0.0121&0.0127&0.0120&0.0106&0.1090&0.0113& 0.0092& 0.0088& 0.0090& 0.0080 \\[0.5ex]  
&&50&0.0089&0.0085&0.0082&0.0078&0.0086&0.0085&0.0087&0.0085&0.0810&0.0121& 0.0097& 0.0110 &0.0072& 0.0083 \\[1.5ex] 
\hline
$t_L(5)$&0.05&10&0.0495&0.0499&0.0525&0.0515&0.0487&0.0512&0.0520&0.0516&0.2181&0.0507& 0.0456 & 0.0509& 0.0468& 0.0464 \\[0.5ex] 
&&25&0.0478&0.0494&0.0526&0.0515&0.0508&0.0517&0.0521&0.0530&0.1770&0.0516& 0.0476& 0.0515& 0.0454& 0.0501  \\[0.5ex] 
&&50&0.0490&0.0491&0.0462&0.0465&0.0478&0.0460&0.0493&0.0478&0.1201&0.0510& 0.0457& 0.0449& 0.0497& 0.0502 \\[1.5ex]
\hline  
&0.01&10&0.0094&0.0093&0.0108&0.0126&0.0098&0.0109&0.0125&0.0133&0.1220&0.0111& 0.0121& 0.0119&  0.0113& 0.0085 \\[0.5ex] 
&&25&0.0092&0.0083&0.0078&0.0091&0.0071&0.0082&0.0088&0.0099&0.1170&0.0110& 0.0116& 0.0081& 0.0097& 0.0083 \\[0.5ex]  
&&50&0.0100&0.0106&0.0117&0.0105&0.0095&0.0094&0.0099&0.0105&0.0860&0.0123& 0.0090& 0.0095& 0.0084& 0.0072 \\[1.5ex]
\hline 
$C_L(0,1)$&0.05&10&0.0455&0.0460&0.0496&0.0516&0.0481&0.0497&0.0506&0.0531& 0.1800&0.0489& 0.0465& 0.0464&  0.0430& 0.0442 \\[0.5ex] 
&&25&0.0471&0.0455&0.0470&0.0481&0.0490&0.0492&0.0486&0.0458&0.1280&0.0516& 0.0479& 0.0444& 0.0511& 0.0442 \\[0.5ex]  
&&50&0.0478&0.0473&0.0474&0.0460&0.0472&0.0455&0.0455&0.0458&0.1060&0.0515& 0.0467& 0.0498 &0.0497& 0.0494 \\[1.5ex]
\hline  
&0.01&10&0.0097&0.0096&0.0099&0.0105&0.0098&0.0098&0.0105&0.0107&0.1665&0.0092& 0.0083& 0.0078& 0.0097& 0.0082 \\[0.5ex]
&&25&0.0130&0.0128&0.0130&0.0131&0.0108&0.0107&0.0104&0.0129&0.1240&0.0110 &0.0107& 0.0079& 0.0079& 0.0104 \\[0.5ex]  
&&50&0.0112&0.0109&0.0111&0.0109&0.0115&0.0116&0.0116&0.0113&0.0910&0.0109& 0.0103 &0.0073& 0.0095& 0.0084 \\[1.5ex]  
\hline 
\end{tabular}%
}
\label{tablep1}%
\end{table}%
\end{landscape}
            
\begin{landscape} 
\begin{table}[!ht]
\centering
\caption{Empirical powers at $5\%$ and $1\%$ level of significance }
\resizebox{21 cm}{!}{
\fontsize{8pt}{11pt}\selectfont
\begin{tabular}{rrrrrrrrrrrrrrrrrrrrrrrrrrr}
\hline
Alt. & $\alpha$& $n$& $T_{n,0.5}^{(1)}$ & $T_{n,1}^{(1)}$ & $T_{n,1.5}^{(1)}$ & $T_{n,3}^{(1)}$ &$T_{n,0.5}^{(2)}$ & $T_{n,1}^{(2)}$&$T_{n,1.5}^{(2)}$ & $T_{n,3}^{(2)}$&$\widehat{\Delta}^J_{W}$  & $\widehat{\Delta}^J_{R}$ & $\widehat{\Delta}^{AJ}_{R}$&  $\widehat{\Delta}^S_{D}$&$\widehat{\Delta}^J_{D}$& $\widehat{\Delta}^{AJ}_{D}$ \\[1.5ex]
    \hline  
$G(2,1)$&0.05&10&0.2633&0.2729&0.3043&0.2728&0.2888&0.2432&0.2325&0.1945&\bf0.4970&0.2946&0.2144&0.2847&0.2411&0.2163   \\[0.5ex]
&&25&0.3539&\bf0.4814&0.4362&0.4162&0.4177&0.3657&0.4125&0.3846&0.4710&0.3273&0.3269&0.4341&0.4122&0.3890 \\[0.5ex]  
&&50&0.3572&0.5430&0.5253&\bf0.5912&0.5693&0.5738&0.5099&0.4740&0.5032&0.4802&0.4951&0.5699&0.3336&0.5085 \\[1.5ex] 
\hline  
&0.01&10&0.1394&0.1307&0.1237&0.0893&0.1168&0.0594&0.0568&0.0383&\bf0.2602&0.1451&0.1279&0.1955&0.1846&0.1439 \\[0.5ex]
&&25&0.1684&\bf0.3553&0.2248&0.2112&0.2398&0.2033&0.1751&0.1496&0.3414&0.1964&0.1667&0.2287&0.2100&0.1873 \\[0.5ex]
&&50&0.1745&\bf0.3948&0.2918&0.3561&0.3016&0.2874&0.3040&0.2751&0.3943&0.3167&0.2997&0.4394&0.3896&0.3168 \\[1.5ex]  
\hline  
$IG(2,1)$&0.05&10&0.2664&0.2853&0.3037&0.2613&0.2816&0.2293&0.2308&0.1925&\bf0.4731&0.2670&0.2355&0.4390&0.3724&0.3688 \\[0.5ex] 
&&25&0.3466&0.4805&0.4561&0.4382&0.4326&0.3822&0.3954&0.3677&0.4740&0.3946&0.3529&\bf0.5297&0.4840&0.4526\\[0.5ex]
&&50&0.3495&0.5275&0.5107&0.5825&0.5434&0.5433&0.5238&0.4935&0.5010&0.5277&0.4991&\bf0.8361&0.7456&0.6332\\[1.5ex]
\hline  
&0.01&10&0.1216&0.1178&0.1342&0.0964&0.1181&0.0585&0.0682&0.0486&0.3651&0.3281&0.2905&0.4508&\bf0.4512&0.4124\\[0.5ex]
&&25&0.1746&0.2752&0.2572&0.2239&0.2342&0.1949&0.1841&0.1537&0.3790&0.3619&0.3461&0.5980&\bf0.6172&0.5731 \\[0.5ex]
&&50&0.1685&0.3369&0.2592&0.3174&0.3562&0.3383&0.3013&0.2745&0.3741&0.6412&0.5477&\bf0.8867&0.8164&0.7102 \\[1.5ex] 
\hline    
$\chi^2(3)$&0.05&10&0.4354&0.4014&0.4144&0.2856&0.3303&0.2078&0.1256&0.1066&0.6587&0.1277&0.1551&\bf0.6017&0.4186&0.3850\\[0.5ex] 
&&25&0.5086&0.5859&0.5029&0.4343&0.3771&0.3144&0.1984&0.2063&\bf0.6620&0.2089&0.3718&0.7123&0.6142&0.5890 \\[0.5ex] 
&&50&0.4861&0.6258&0.4485&0.5337&0.3456&0.4302&0.2256&0.2731&0.7031&0.4974&0.4377&\bf0.9871&0.8897&0.8366 \\[1.5ex] 
\hline 
&0.01&10&0.2384&0.2198&0.2174&0.1397&0.1759&0.0636&0.0442&0.0187&\bf0.5531&0.2514&0.2105&0.3896&0.2215&0.1988\\[0.5ex]
&&25&0.3043&0.3567&0.3195&0.2336&0.2247&0.1692&0.0634&0.0452&0.5642&0.3140&0.4193&\bf0.7228&0.3461&0.2967\\[0.5ex]
&&50&0.2843&0.4244&0.2865&0.3163&0.2032&0.2312&0.1119&0.1402&0.6064&0.5446&0.5011&\bf0.9630&0.6413&0.5877\\[1.5ex]
\hline  
$L(2)$&0.05&10&\bf0.7812&0.6269&0.6541&0.4537&0.4833&0.3063&0.0795&0.0672&0.7743&0.4497&0.4092&0.6189&0.5986&0.5685\\[0.5ex] 
&&25&\bf0.9931&0.9648&0.9519&0.8238&0.7676&0.7033&0.2172&0.4334&0.9924&0.6121&0.5946&0.8645&0.7491&0.7367\\[0.5ex] &&50&\bf1.0000&\bf1.0000&0.9967&0.9890&0.9377&0.9596&0.5881&0.8621&\bf1.0000&0.9172&0.7793&0.9682&0.9549&0.8964\\[1.5ex] 
\hline    
&0.01&10&0.4856&0.3988&0.4188&0.2648&0.2979&0.1564&0.0325&0.0274&\bf0.6860&0.2845&0.2160&0.4183&0.3641&0.3150\\[0.5ex]
&&25&0.9095&0.7966&0.7944&0.5537&0.5378&0.4026&0.0355&0.0645&\bf0.9678&0.5611&0.5410&0.6978&0.6212&0.5679\\[0.5ex] 
&&50&0.9953&0.9958&0.9747&0.9156&0.7924&0.7848&0.2707&0.5956&\bf1.0000&0.8775&0.7594&0.9467&0.9211&0.8788\\[1.5ex]
\hline  
$W(2,1)$&0.05&10&0.3532&0.5987&0.5365&0.5847&0.6029&0.5752&0.5323&0.5007&0.6124&0.5478&0.4908&0.7044&\bf0.6609&0.5861 \\[0.5ex] 
&&25&0.9323&0.9992&0.9996&0.9994&\bf0.9997&0.9997&0.9607&0.9901&0.9041&0.7946&0.7455&0.9767&0.9588&0.8897\\[0.5ex]
&&50&0.9991&\bf1.0000&\bf1.0000&\bf1.0000&\bf1.0000&\bf1.0000&\bf1.0000&\bf1.0000&0.9921&0.9987&0.9673&\bf1.0000&0.9998&0.9768\\[1.5ex] 
\hline     
&0.01&10&0.0889&0.1981&0.1159&0.2429&0.1779&0.2552&0.2856&0.2952&0.4812&0.4182&0.3974&\bf0.5198&0.4866&0.4598\\[0.5ex] 
&&25&0.7598&0.9695&0.9677&\bf0.9743&0.9614&0.9362&0.7038&0.6783&0.7941&0.6511&0.5992&0.8340&0.7980&0.7609\\[0.5ex]
&&50&0.9905&\bf1.0000&\bf1.0000&\bf1.0000&\bf1.0000&\bf1.0000&0.9981&\bf1.0000&0.9681&0.8297&0.7948&0.9587&0.9453&0.9044\\[1.5ex]
\hline
    \end{tabular}%
    }
  \label{tablep2}%
\end{table}%
\end{landscape}

\begin{landscape} 
\begin{table}[!ht]
 \centering
  \caption{Empirical powers at $5\%$ and $1\%$ level of significance }
    \resizebox{21 cm}{!}{
    \fontsize{8pt}{11pt}\selectfont
    \begin{tabular}{rrrrrrrrrrrrrrrrrrrrrrrrrrr}
    \hline
   Alt. & $\alpha$& $n$& $T_{n,0.5}^{(1)}$ & $T_{n,1}^{(1)}$ & $T_{n,1.5}^{(1)}$ & $T_{n,3}^{(1)}$ &$T_{n,0.5}^{(2)}$ & $T_{n,1}^{(2)}$&$T_{n,1.5}^{(2)}$ & $T_{n,3}^{(2)}$&$\widehat{\Delta}^J_{W}$  & $\widehat{\Delta}^J_{R}$ & $\widehat{\Delta}^{AJ}_{R}$&  $\widehat{\Delta}^S_{D}$&$\widehat{\Delta}^J_{D}$& $\widehat{\Delta}^{AJ}_{D}$ \\[1.5ex]
   \hline  
$M(2)$&0.05&10&0.6904&0.6973&0.7387&0.6389&0.6975&0.5897&0.5632&0.4491&\bf0.8290&0.6410&0.5941&0.7455&0.7264&0.6972\\[0.5ex] 
&&25&0.8458&\bf0.9295&0.9091&0.8785&0.8776&0.8034&0.7945&0.7647&0.8761&0.8413&0.7746&0.8970&0.8164&0.7833 \\[0.5ex] 
&&50&0.8838&\bf0.9714&0.9581&0.9706&0.9489&0.9474&0.9094&0.8999&0.9034&0.9688&0.9128&0.9987&0.9712&0.9568\\[1.5ex] 
\hline   
&0.01&10&0.5083&0.4947&0.5093&0.4063&0.4722&0.3162&0.2692&0.1703&\bf0.7721&0.4487&0.4330&0.5696&0.5219& 0.4695\\[0.5ex] 
&&25&0.6564&0.7915&0.7713&0.6905&0.7423&0.6406&0.6118&0.5566&0.8012&0.7409&0.6866&\bf0.8456&0.8297&0.7891 \\[0.5ex]  
&&50&0.7294&\bf0.9179&0.8695&0.8907&0.8549&0.8398&0.7858&0.7504&0.8601&0.9132&0.8736&0.8609&0.8459&0.8270 \\[1.5ex]  
\hline 
$P(1)$&0.05&10&\bf1.0000&\bf1.0000&\bf1.0000&\bf1.0000&\bf1.0000&\bf1.0000&0.2647& 0.1434&\bf1.0000&0.7984&0.6980&0.9136&0.8571&0.8362 \\[0.5ex]
&&25&\bf1.0000&\bf1.0000&\bf1.0000&\bf1.0000&\bf1.0000&\bf1.0000&\bf1.0000&\bf1.0000&\bf1.0000&0.9761&0.9644&0.9780&0.9722&0.9455\\[0.5ex] 
&&50&\bf1.0000&\bf1.0000&\bf1.0000&\bf1.0000&\bf1.0000&\bf1.0000&\bf1.0000&\bf1.0000&\bf1.0000&\bf1.0000&\bf1.0000&\bf1.0000&\bf1.0000&0.9984\\[1.5ex] 
\hline   
&0.01&10&\bf1.0000&\bf1.0000&\bf1.0000&0.8119&\bf1.0000&0.2622&0.0864&0.0588&0.9981&0.7462&0.6897&0.9784&0.9551& 0.9422 \\[0.5ex] 
&&25&\bf1.0000&\bf1.0000&\bf1.0000&\bf1.0000&\bf1.0000&\bf1.0000&\bf1.0000&\bf1.0000&\bf1.0000&0.9465&0.9295&\bf1.0000&0.9960&0.9875 \\[0.5ex]
&&50&\bf1.0000&\bf1.0000&\bf1.0000&\bf1.0000&\bf1.0000&\bf1.0000&\bf1.0000&\bf1.0000&\bf1.0000&\bf1.0000&0.9982&\bf1.0000&\bf1.0000&\bf1.0000 \\[1.5ex] 
\hline 
$IB(0.6,1)$&0.05&10&\bf1.0000&\bf1.0000&\bf1.0000&\bf1.0000&\bf1.0000&\bf1.0000&0.3792& 0.1751&\bf1.0000&0.7844&0.5980&0.9894&0.9580&0.9367 \\[0.5ex] 
&&25&\bf1.0000&\bf1.0000&\bf1.0000&\bf1.0000&\bf1.0000&\bf1.0000&\bf1.0000&\bf1.0000&\bf1.0000&0.9896&0.9588&\bf1.0000&0.9986&0.9899 \\[0.5ex] 
&&50&\bf1.0000&\bf1.0000&\bf1.0000&\bf1.0000&\bf1.0000&\bf1.0000&\bf1.0000&\bf1.0000&\bf1.0000&\bf1.0000&\bf1.0000&\bf1.0000&\bf1.0000&\bf1.0000 \\[1.5ex] 
\hline   
&0.01&10&\bf1.0000&\bf1.0000&\bf1.0000&0.9941&\bf1.0000&0.6415&0.0937&0.0394&\bf1.0000&0.8746&0.8133&0.9148&0.8977&0.8654 \\[0.5ex]
&&25&\bf1.0000&\bf1.0000&\bf1.0000&\bf1.0000&\bf1.0000&\bf1.0000&\bf1.0000&\bf1.0000&\bf1.0000&0.9986&0.9689&\bf1.0000&0.9995&0.9897 \\[0.5ex]
&&50&\bf1.0000&\bf1.0000&\bf1.0000&\bf1.0000&\bf1.0000&\bf1.0000&\bf1.0000&\bf1.0000&\bf1.0000&\bf1.0000&\bf1.0000&\bf1.0000&\bf1.0000&\bf1.0000\\[1.5ex]  
\hline  
$B(1,0.1)$&0.05&10&0.9998&\bf1.0000&\bf1.0000&\bf1.0000&\bf1.0000&\bf1.0000&0.4313&0.2397&0.9984&0.9791&0.9356&0.9981&0.9568&0.9154 \\[0.5ex] 
&&25&\bf1.0000&\bf1.0000&\bf1.0000&\bf1.0000&\bf1.0000&\bf1.0000&\bf1.0000&\bf1.0000&\bf1.0000&\bf1.0000&0.9682&\bf1.0000&0.9891&0.9779 \\[0.5ex]
&&50&\bf1.0000&\bf1.0000&\bf1.0000&\bf1.0000&\bf1.0000&\bf1.0000&\bf1.0000&\bf1.0000&\bf1.0000&\bf1.0000&0.9999&\bf1.0000&\bf1.0000&0.9994 \\[1.5ex] 
\hline   
&0.01&10&0.9985&\bf1.0000&\bf1.0000&0.9311&\bf1.0000&0.2729&0.1568&0.1025&\bf1.0000&0.9184&0.9176&0.9778&0.9544& 0.9423 \\[0.5ex]
&&25&\bf1.0000&\bf1.0000&\bf1.0000&\bf1.0000&\bf1.0000&\bf1.0000&\bf1.0000&\bf1.0000&\bf1.0000&0.9874&0.9668&\bf1.0000&0.9897&0.9680 \\[0.5ex]
&&50&\bf1.0000&\bf1.0000&\bf1.0000&\bf1.0000&\bf1.0000&\bf1.0000&\bf1.0000&\bf1.0000&\bf1.0000&\bf1.0000&0.9986&\bf1.0000&\bf1.0000&0.9980 \\[1.5ex] 
\hline 
$TP(1)$&0.05&10&0.9996&\bf1.0000&\bf1.0000&\bf1.0000&\bf1.0000&\bf1.0000&0.9865&0.9073&\bf1.0000&0.9413&0.9358&0.9997&0.9891& 0.9683 \\[0.5ex] 
&&25&\bf1.0000&\bf1.0000&\bf1.0000&\bf1.0000&\bf1.0000&\bf1.0000&\bf1.0000&\bf1.0000&\bf1.0000&\bf1.0000&0.9789&0.9999&\bf1.0000&0.9983 \\[0.5ex] 
&&50&\bf1.0000&\bf1.0000&\bf1.0000&\bf1.0000&\bf1.0000&\bf1.0000&\bf1.0000&\bf1.0000&\bf1.0000&\bf1.0000&\bf1.0000&\bf1.0000&\bf1.0000&\bf1.0000 \\[1.5ex]     
\hline   
&0.01&10&0.9991&\bf1.0000&\bf1.0000&0.9776&\bf1.0000&0.8659&0.8088&0.6812&\bf1.0000&0.9689&0.9510&0.9784&0.9688& 0.9576 \\[0.5ex]
&&25&\bf1.0000&\bf1.0000&\bf1.0000&\bf1.0000&\bf1.0000&\bf1.0000&\bf1.0000&\bf1.0000&\bf1.0000&0.9990&0.9983&\bf1.0000&0.9998&0.9689 \\[0.5ex] 
&&50&\bf1.0000&\bf1.0000&\bf1.0000&\bf1.0000&\bf1.0000&\bf1.0000&\bf1.0000&\bf1.0000&\bf1.0000&\bf1.0000&\bf1.0000&\bf1.0000&\bf1.0000&\bf1.0000 \\[1.5ex]  
\hline
    \end{tabular}%
    }
  \label{tablep3}%
\end{table}%
\end{landscape}
 
{In addition to empirical type I error and power, we assessed the computational efficiency of the proposed Fourier-based tests relative to existing methods. For each test, the average CPU time per Monte Carlo replication (in seconds) was measured using \texttt{proc.time()} in \texttt{R}. Table \ref{tab:com1} presents a conceptual comparison of all tests for log-symmetric distributions, reporting the underlying principle, type of test statistic, asymptotic null distribution, average computational time, and key computational features. The results indicate that the proposed Fourier-based tests are considerably faster than warp-speed bootstrap and empirical likelihood-based methods, as it relies on closed-form kernel expressions and avoids resampling or constrained optimization, thereby substantially reducing computation time.}

\begin{table}[ht!]
\centering
\caption{{Comparison of the proposed and existing methods for log-symmetric distributions}}
\label{tab:com1}
\resizebox{17cm}{!}{
\begin{tabular}{p{2cm} p{3cm} p{3.1cm}p{2.4cm} p{2.2cm} p{6.5cm}}
\hline
\textbf{Test} & \textbf{Underlying principle} & \textbf{Test statistic type} & \textbf{Asymptotic null dist.} &\textbf{Avg. time per MC repl. (s)} & \textbf{Computational features} \\[1ex]
\hline
$T_{n,a}^{(1)}$
& Weighted characteristic function (CF) approach
& Integrated squared deviation of empirical CF
& Weighted $\chi^2_1$
&0.0512
& Simple to compute; test statistic is computed using closed-form kernel expressions; tuning parameter $a$ controls sensitivity \\[1ex]
\hline
$T_{n,a}^{(2)}$
& Weighted CF approach
& Quadratic functional of empirical CF
& Weighted $\chi^2_1$
& 0.0684 
& Computationally efficient; significantly faster than bootstrap-based procedures for moderate and large samples; good small-sample behavior \\[1ex]
\hline
$\widehat{\Delta}^{J}_{W}$ \citep{Anjana21072025}
& Probability-weighted moment characterization
& $U$-statistic + JEL
& $\chi^2_1$
& 0.2250
& Moderate computational cost; depends on order statistics; fails to attain the nominal significance level for small samples. \\[1ex]
\hline
$\widehat{\Delta}^{J}_{R}$ \citep{avhad2025jackknife}
& Ratios of consecutive order statistics
& $U$-statistic + JEL 
& $\chi^2_1$
& 76.4635
& Computationally intensive due to combinatorial summation; involve higher-order $U$-statistics. \\[1ex]
\hline
 $\widehat{\Delta}^{AJ}_{R}$ \citep{avhad2025jackknife}
& Ratios of consecutive order statistics
& $U$-statistic + AJEL
& $\chi^2_1$
& 78.4635
& Still substantially slower than closed-form test statistics; improves numerical stability but at the cost of extra computation. \\[1ex]
\hline
$\widehat{\Delta}^{S}_{D}$ \citep{avhad2025family}
& Characterization of log-symmetry
& $U$-statistic (asymptotic test)
& Normal 
&15.1342
& Fast; relies on asymptotic critical values; computational cost grows linearly with the number of bootstrap samples. \\[1ex]
\hline
$\widehat{\Delta}^{J}_{D}$ \citep{avhad2025family}
& Characterization of log-symmetry
& $U$-statistic + JEL  
&  $\chi^2_1$
& 26.5776
& Moderate computational cost, improved small-sample accuracy \\[1ex]
\hline
  $\widehat{\Delta}^{AJ}_{D}$ \citep{avhad2025family}
& Characterization of log-symmetry
&$U$-statistic + AJEL
& $\chi^2_1$
& 28.6376
& Compared to JEL, AJEL has a slightly higher computational burden while retaining good small-sample accuracy. \\[1ex]  
\hline
\end{tabular}}
\end{table}

\section{Real data analysis}\label{sec5}
\noindent In this section, we demonstrate the application of each proposed test to real data sets. 
 
\subsection*{Illustration I.}
\noindent We consider a data set representing the breakdown times of an insulating fluid, originally analyzed by Reath {et al}. \cite{reath2018improved} and presented in Table \ref{table4}. To obtain preliminary insights into the underlying distributional structure, Figure \ref{fig:ex1} suggests that the data exhibit patterns consistent with log-symmetric behavior. {The empirical cumulative distribution function (ECDF) plot indicates that the empirical distribution closely follows the fitted log-logistic distribution, suggesting that the log-symmetric model provides an adequate fit to the time-to-breakdown data.} In our study, this data set is employed to assess the log-symmetry of the underlying distribution. 

\begin{table}[!ht]
    \centering
    \caption{The time to breakdown of an insulating fluid}\label{table4} 
    \resizebox{13.5cm}{!}{\fontsize{8pt}{11pt}\selectfont
    \begin{tabular}{c c c c c c c c c c c c c c c c c c c c c c c}
    \hline
     0.96&4.15&0.19&0.78&8.01&31.75&7.35&6.50&8.27&33.91\\
     32.52&3.16&4.85&2.78&4.67&1.31&12.06& 36.71&72.89\\
    \hline
    \end{tabular}}
\end{table}

\begin{figure}[!ht]
    \centering
    \includegraphics[width=16.5cm,height=9cm]{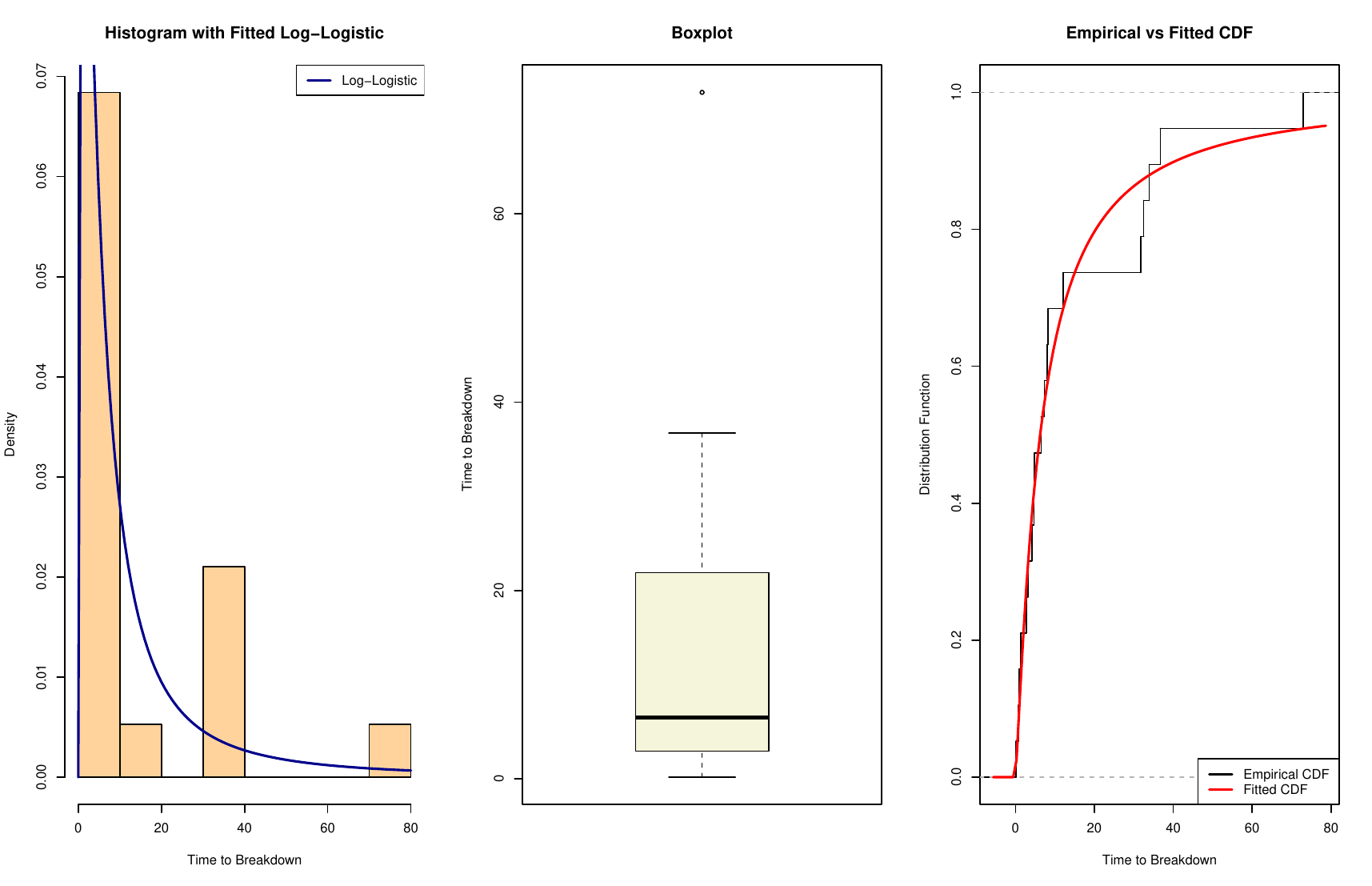}
    \caption{Histogram, boxplot, and empirical distribution function of insulating fluid data set}
    \label{fig:ex1}
\end{figure}
 
\subsection*{Illustration II.}
\noindent The data set, obtained from Qiu and Jia \cite{qiu2018extropy} and presented in Table \ref{table5}, has previously been modeled using an inverse asymmetric Gaussian distribution. In this study, it is used to demonstrate the applicability of the proposed test statistic in assessing distributional log-symmetry. Figure \ref{fig:ex2} shows the corresponding boxplot, providing a visual summary of the skewness and dispersion. {The ECDF plot shows good agreement between the empirical and fitted distribution functions, implying that the log-symmetric assumption is reasonable for the active repair time data.} The proposed tests are then applied to assess whether the distribution exhibits symmetry on the logarithmic scale.

\begin{table}[!ht]
    \centering
    \caption{Active repair times (in hours) for an airborne communication transceiver}\label{table5}
    \resizebox{14 cm}{!}{\fontsize{9pt}{11pt}\selectfont
    \begin{tabular}{p{0.6cm} p{0.6cm} p{0.6cm} p{0.6cm} p{0.6cm}p{0.6cm}  p{0.6cm}p{0.6cm} p{0.6cm} p{0.6cm} p{0.6cm} p{0.6cm} }
    \hline
    0.2& 0.3& 0.5& 0.5& 0.5& 0.5& 0.6& 0.6& 0.7& 0.7& 0.7& 0.8 \\[0.5ex]
    0.8& 1.0& 1.0& 1.0& 1.0& 1.1& 1.3& 1.5& 1.5& 1.5& 1.5& 2.0 \\[0.5ex]
    2.0& 2.2& 2.5& 3.0& 3.0& 3.3& 3.3& 4.0& 4.0& 4.5& 4.7& 5.0\\[0.5ex]
    5.4& 5.4& 7.0& 7.5& 8.8& 9.0& 10.3& 22.0& 24.5  \\
   \hline
   \end{tabular}}
\end{table}

\begin{figure}[!ht]
    \centering
    \includegraphics[width=16.5cm,height=9cm]{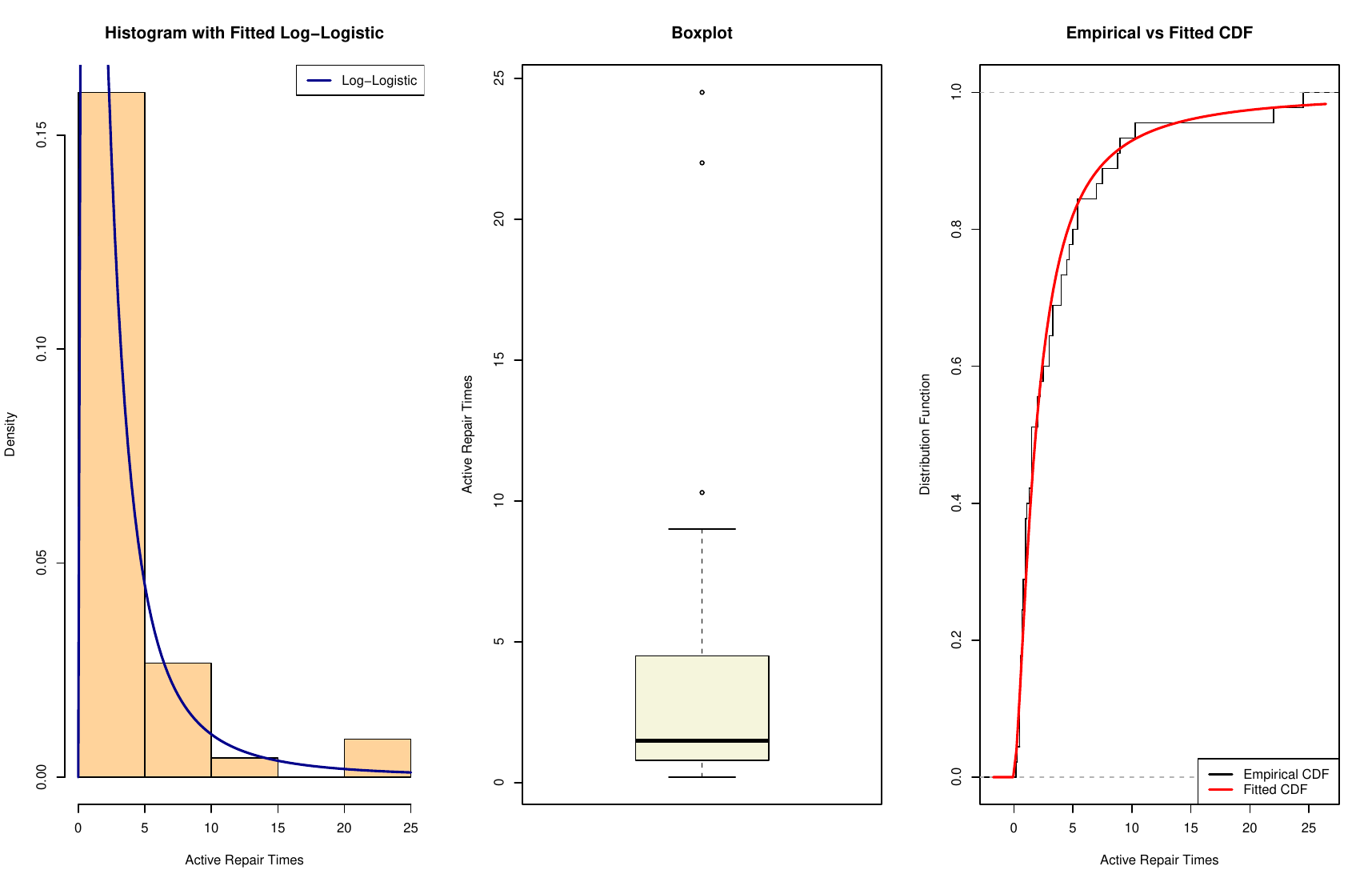}
    \caption{Histogram, boxplot, and empirical distribution function of airborne communication transceiver data set}
    \label{fig:ex2}
\end{figure}

\begin{table}[!ht]
  \centering    
  \caption{$p$-values associated with the various tests considered}
  \resizebox{16.5cm}{!}{
  \fontsize{9pt}{12pt}\selectfont
  \begin{tabular}{cccccccccccccccc}
    \hline
   \ Test &$T_{n,0.5}^{(1)}$ & $T_{n,1}^{(1)}$ & $T_{n,1.5}^{(1)}$ & $T_{n,3}^{(1)}$ &$T_{n,0.5}^{(2)}$ & $T_{n,1}^{(2)}$&$T_{n,1.5}^{(2)}$ & $T_{n,3}^{(2)}$&$\widehat{\Delta}^J_{W}$  & $\widehat{\Delta}^J_{R}$ & $\widehat{\Delta}^{AJ}_{R}$&  $\widehat{\Delta}^S_{D}$&$\widehat{\Delta}^J_{D}$& $\widehat{\Delta}^{AJ}_{D}$ \\[1.5ex]
    \hline
     Illustration 1 &0.979&0.984&0.919&0.325&0.991 &0.927&0.807 &0.473&0.480 &0.993&0.984&1.000&0.998&1.000    \\[1.5ex]  
    \cline{1-15}
    Illustration 2&0.994 &0.988&0.903&0.691&0.994 &0.981& 0.933&0.880 &0.390&0.990&0.996&1.000&1.000&1.000 \\[1.5ex]  
   \hline
  \end{tabular}
  }
  \label{tab:results}
\end{table}

In order to find the critical values for the tests considered, we use warp-speed bootstrap with $B=1000$ replications. To generate the bootstrap samples, the null distribution is considered to be log-logistic. 
Table~\ref{tab:results} reports the $p$-values obtained from the tests described in Section~\ref{sec4}. {The obtained $p$-values show how well the assumed model matches the observed data. Large $p$-values, or test statistics that are smaller than their critical values, mean that the differences between the data and the model are small and can be explained by random variation.} For all the considered procedures, the resulting $p$-values exceed the nominal $5\%$ significance level, indicating insufficient evidence to reject the null hypothesis of log-symmetry. Consequently, the log-symmetric distribution provides a satisfactory and interpretable fit for the observed data in the considered application.


\section{Concluding remarks}\label{sec6} 
\noindent In this paper, based on the characterization introduced by Ahmadi and Balakrishnan \cite{ahmadi2024characterizations}, we proposed a novel weighted $L^2$-type goodness-of-fit tests for the class of log-symmetric distributions, constructed using the empirical characteristic function. The asymptotic behavior of the test statistic is thoroughly investigated. The finite-sample performance of the proposed methods is assessed through extensive Monte Carlo simulations. The results demonstrate that the weighted $L^2$-type tests effectively control the nominal significance level, show consistently high empirical power under various competing alternatives, and offer computational efficiency compared to existing methods. Among all the tests, $T_{n,1}^{(1)}$ is the most effective based on its overall performance. It achieves high power and accurately maintains the nominal significance level, making it suitable for practical use. Finally, the applicability of the proposed methods is illustrated using real-world data sets.

We conclude the paper by outlining several directions for future research. First, based on similar characterizations, more robust and powerful tests for log-symmetric distributions can be developed and systematically compared with existing procedures. Second, since the performance of tests based on the empirical characteristic function depends on the choice of the tuning parameter $a$, a deeper understanding of its role in determining test power remains an important open problem. Third, as the test statistic \eqref{Tna} is weighted and non-linear, the standard empirical likelihood framework is not directly applicable; however, weighted empirical likelihood approaches with alternative weight functions may be explored (see, Glenn and Zhao \cite{glenn2007weighted}). {In addition, natural extensions include multivariate log-symmetry testing using joint characteristic functions or copula-based dependence structures, as well as regression settings incorporating covariates information, such as accelerated failure time and scale-location models.} 



\begin{appendix}
\section{Algorithms}\label{A1}
\begin{algorithm}[H]
\caption{{Computation of the test statistic $T_{n,a}^{(\cdot)}$}}
\label{alg:Tna}
\begin{algorithmic}[1]
\Require Sample $X_1,\ldots,X_n$ from a positive distribution, tuning parameter $a>0$
\Ensure Value of the test statistic $T_{n,a}$

\State Sort the sample to obtain $X_{(1)} \le \cdots \le X_{(n)}$
\State Compute weights
\[
\beta_{m,n} = \left(\frac{n-m}{n-1}\right)^{n-1}, ~~\text{and}~~
\gamma_{m,n} = \left(\frac{m-1}{n-1}\right)^{n-1}, \quad m=1,\ldots,n
\]

\State Initialize $S \gets 0$
\For{$m=1$ to $n$}
  \For{$\ell=1$ to $n$}
    \State Compute kernel terms
    \[
    K_1 = k_a(X_{(m)}-X_{(\ell)}), \quad
    K_2 = k_a(X_{(m)}-X_{(\ell)}^{-1}), ~~\text{and}~~
    K_3 = k_a(X_{(m)}^{-1}-X_{(\ell)}^{-1}),
    \]
    where, $k_a(\cdot)$ denotes the kernel function. In particular,
\[
k_a(t)=\frac{a}{t^2+a^2} \quad \text{(Laplace kernel)},
~~\text{and}~~
k_a(t)=\exp\!\left(-\frac{t^2}{4a}\right) \quad \text{(Gaussian kernel)}.
\]
    \State Update
    \[
    S \gets S + \beta_{m,n} \beta_{\ell,n} K_1 - 2 \gamma_{m,n} \gamma_{\ell,n} K_2 +  \beta_{m,n} \gamma_{\ell,n} K_3
    \]
  \EndFor
\EndFor

\State Set
\[
T_{n,a}^{(\cdot)} =
\begin{cases}
\dfrac{1}{n} S, & \text{for Laplace kernel},\\[1ex]
\sqrt{\dfrac{\pi}{a}}\,\dfrac{1}{n} S, & \text{for Gaussian kernel.}
\end{cases}
\]
\State \Return $T_{n,a}^{(\cdot)}$
\end{algorithmic}
\end{algorithm}

\begin{algorithm}[H]
 \caption{Bootstrap algorithm.}
 \begin{algorithmic}[ht!]
\State \textbf{Input:} Sample size $n$, number of Monte Carlo replications $(MC)$
\State  $\bullet$ Draw a sample of size $n$, say $X_1, \dots, X_n$ is simulated from a distribution.
\State $\bullet$ Compute the test statistic $T = T^{(\cdot)}_{n,a}(X_1, \dots, X_n)$. 
 \For{$j = 1$ to $MC$}
    \State $\bullet$ Generate a bootstrap sample $X_1^*, \dots, X_{n}^*$ from null distributions.
    \State $\bullet$ Compute the test statistic on the bootstrap sample: $T_j^* = T^{(\cdot)}_{n,a}(X_1^*, \dots, X_{n}^*)$.
    \State $\bullet$ Store $T_j = T$.
\EndFor
\State $\bullet$ Compute the empirical power:
  $\dfrac{1}{MC} \sum\limits_{j=1}^{MC} \mathrm{I}(T_j > T_j^*)$.
\end{algorithmic}
\end{algorithm}  
\end{appendix}


\end{document}